\documentclass[draftcls,onecolumn]{IEEEtran}

\usepackage{amsmath, amsfonts,amssymb, latexsym,epsfig,color,rotating,paralist,times,float,subfigure,algorithm,verbatim,amsthm}

\makeatletter
\def\nl#1#2{\begingroup
    #2%
    \def\@currentlabel{#2}%
    \phantomsection\label{#1}\endgroup
}
\makeatother

{\ensuremath{
\begin{empheq}[box=\fbox]{align}
{#1}
\end{empheq}
}}

\newtheorem{theorem}            {Theorem}[section]
\newtheorem{conjecture}            {Conjecture}[section]
\newtheorem{corollary}          [theorem]{Corollary}

\newtheorem{definition}         [theorem]{Definition}

{
\begin{mdframed}
\par\noindent\textbf{#1:}\begin{rmfamily}\noindent}%
{\end{rmfamily}
\end{mdframed}
}

\newcommand{\I}{\Pi}

\newcommand{\prob}{\mathbb{P}}
\newcommand{\E}                 {\Bbb{E}}

\renewcommand{\P}                 {\Bbb{P}}

\newcommand{\obs}{y}

\newcommand{\state}{x}
\newcommand{\statespace}{\mathcal{X}}
\newcommand{\obspace}{\mathcal{Y}}
\newcommand{\statedim}{X}
\newcommand{\obsdim}{{Y}}

\newcommand{\fun}{\phi}

\newcommand{\oprob}{B}
\newcommand{\tp}{P}

\newcommand{\belief}{\pi}

\newcommand{\Belief}{\Pi(\statedim)}

\newcommand{\sigs}{\sigma}

\newcommand{\ole}{\stackrel{\text{defn}}{=}}

\newcommand{\gr}{\geq_r}

\newcommand{\filterd}{\sigma}

\newcommand{\filter}{T}

\newcommand{\argmin}{\operatornamewithlimits{argmin}}

\newcommand{\reals}{{\rm I\hspace{-.07cm}R}}

\newcommand{\beq}{\begin{equation}}
\newcommand{\eeq}{\end{equation}}
\newcommand{\nn}{\nonumber}

\newcommand{\p}{\prime}

\newcommand{\one}{\mathbf{1}}
\newcommand{\ones}{\mathbf{1}}

\newcommand{\diag}{\textnormal{diag}}

\newcommand{\ltwo}{\log_2}

\newcommand{\cost}{c}

\newcommand{\Cost}{C}
\newcommand{\nlcost}{d}
\newcommand{\nlCost}{D}

\newcommand{\yi}{y^{(1)}}
\newcommand{\yii}{y^{(2)}}

\newcommand{\action}{u}

\newcommand{\actionspace}{\,\mathcal{U}}
\newcommand{\actiondim}{U}

\newcommand{\discount}{\rho}

\newcommand{\region}{\mathcal{R}}

 \newcommand{\Ep}{\E_{\policy}}

\newcommand{\policy}{\mu}

\newcommand{\optpolicy}{\policy^*}

\newcommand{\valuef}{V}

\newcommand{\info}{\mathcal{I}}

\renewcommand{\time}{k}

\newcommand{\bd}{\succeq_{\mathcal{B}}}
\newcommand{\aB}{R}

\newcommand{\valueaction}{Q}

\newcommand{\optvalue}{V}

\newcommand{\filternorm}{\sigma}

\newcommand {\policyl} {\underline{\mu}}

\newcommand{\Pp}{\P_\mu}

\newcommand{\barray}{\begin{array}{ll}}
\newcommand{\earray}{\end{array}}

\newcommand{\valuer}{W}
\newcommand{\rbelief}{\alpha}

\begin{document}

\title{POMDP Structural Results for Controlled Sensing}
\author{Vikram Krishnamurthy
\thanks{V. Krishnamurthy is
 with the Department of Electrical and Computer
Engineering, Cornell University, USA. 
(email:  vikramk@cornell.edu).}}

\maketitle


%


\section{Introduction}
Structural results for POMDPs are important since solving POMDPs numerically are typically intractable. Solving a classical POMDP is known to be PSPACE-complete
\cite{PT87}.
Moreover, in controlled sensing problems \cite{Kri02,KD09,EKN05}, it is often necessary to use POMDPs that are 
nonlinear in the belief state in order to model the uncertainty in the state estimate. (For example, the variance of the state estimate is a quadratic function of the belief.)
In such cases, there is no finite dimensional characterization of the optimal POMDP policy even for a finite horizon.

The seminal papers \cite{Lov87,Rie91,RZ94} give sufficient conditions on the costs, transition provabilities and observation probabilities so  that the value function of a POMDP is monotone with respect to 
the monotone likelihood ratio  (MLR) order (and more generally the multivariate TP2 order).
These papers then use this monotone  result  to show that the optimal policy
can be lower bounded by a myopic policy. 
Our recent works \cite{Kri16,KP15} relax the conditions on the transition matrix to construct myopic lower and upper bounds.

\section {The Partially Observed Markov Decision Process} \label{sec:pomdp}
For notational convenience, we consider  a discrete time, infinite horizon discounted cost POMDP. A   discrete time Markov chain  evolves on the  state space $\statespace = \{1,2,\ldots, \statedim\}$. Denote the
action space  as $\actionspace = \{1,2,\ldots,\actiondim\}$ and observation space as $\obspace$. For discrete-valued observations $\obspace = \{1,2,\ldots,\obsdim\}$ and for continuous observations $\obspace \subset \reals$.

Let
$\Belief = \left\{\belief: \belief(i) \in [0,1], \sum_{i=1}^\statedim \belief(i) = 1 \right\}$ denote the belief space of $\statedim$-dimensional probability vectors.  For stationary policy  $\policy: \Belief \rightarrow \actionspace$,
 initial belief  $\belief_0\in \Belief$,  discount factor $\discount \in [0,1)$, define the  discounted cost:
\begin{align}\label{eq:discountedcost}
J_{\policy}(\belief_0) = \Ep\left\{\sum_{\time=0}^{\infty} \discount ^{\time} \cost_{\policy(\belief_\time)}^\p\belief_\time\right\}.
\end{align}
Here $\cost_\action = [\cost(1,\action),\ldots,\cost(\statedim,\action)]^\p$, $u\in \actionspace$ is the cost vector for each action, and the belief state evolves as
$\belief_{k} = \filter(\belief_{k-1},\obs_k,\action_k)$ where
\begin{align}  \filter\left(\belief,\obs,\action\right) = \cfrac{\oprob_{\obs} (\action)\, \tp^\p(\action)\belief}{\filterd\left(\belief,\obs,\action\right)} , \quad
\filterd\left(\belief,\obs,\action\right) = \one_{\statedim}'\oprob_{\obs}(\action) \tp^\p(\action)\belief, \quad
\oprob_{\obs}(\action) = \diag\{\oprob_{1,\obs}(\action),\cdots,\oprob_{\statedim,\obs}(\action)\}. \label{eq:information_state}
\end{align}
Here  $\one_{\statedim}$ represents a $\statedim$-dimensional vector of ones,
$ \tp(\action) = \left[\tp_{ij}(\action)\right]_{\statedim\times\statedim}$
$ \tp_{ij}(\action) = \prob(\state_{\time+1} = j | \state_\time = i, \action_\time = a )$ denote the transition probabilities,
 $\oprob_{\state\obs}(\action) = \prob(\obs_{\time+1} = \obs| \state_{\time+1} = \state, \action_{\time} = \action)$ when $\obspace$ is finite,
 or  $\oprob_{\state\obs}(\action)$ is the conditional probability density function when $\obspace \subset \reals$.

The aim is to compute the optimal  stationary policy $\optpolicy:\Belief \rightarrow \actionspace$ such that
$J_{\optpolicy}(\belief_0) \leq J_{\policy}(\belief_0)$ for all $\belief_0 \in \Belief$.
Obtaining the optimal policy  $\optpolicy$ is equivalent to solving
 Bellman's  dynamic programming equation:
$ \optpolicy(\belief) =  \underset{\action \in \actionspace}\argmin~ Q(\belief,\action)$, $J_{\optpolicy}(\belief_0) = \valuef(\belief_0)$, where
\begin{equation}
\valuef(\belief)  = \underset{\action \in \actionspace}\min ~Q(\belief,\action), \quad
  Q(\belief,\action) =  ~\cost_\action^\prime\belief + \discount\sum_{\obs \in \obsdim} \valuef\left(\filter\left(\belief,\obs,\action\right)\right)\filterd \left(\belief,\obs,\action\right). \label{eq:bellman}
\end{equation}

Since  $\Belief$ is continuum, Bellman's equation \eqref{eq:bellman} does not translate into practical solution methodologies as  the value function $\valuef(\belief)$ needs to be evaluated at each $\belief \in \Belief$.

\subsection{POMDPs in Controlled Sensing}
In controlled sensing, to incorporate
 uncertainty of  the state estimate, we generalize the above POMDP to consider costs that are nonlinear in the belief.
Consider the following   instantaneous cost at each time $k$:
 $$ \cost(\state_k,\action_k) + \nlcost(\state_k,\belief_k,\action_k), \quad \action_k \in \actionspace = \{1,2,\ldots,\actiondim\}. $$
  (i) {\em Sensor Usage Cost}:
$\cost(\state_k,\action_k)$ denotes the instantaneous cost of using sensor  $\action_k$  at time $k$ when the
 Markov chain is in state $\state_k$. \\
 (ii) {\em Sensor Performance Loss}:  $\nlcost(\state_k,\belief_k,\action_k)$ models the
performance loss when using sensor $\action_k$.  This loss is modeled as an  explicit function of the belief state $\belief_k$  to capture the uncertainty in the state estimate. 

Typically there is trade off between the sensor usage cost and performance loss.
 Accurate sensors have high usage cost but small performance loss.

Then in terms of the belief state,  the instantaneous cost can be expressed as
\beq \label{eq:nlcoststart}
\begin{split} \Cost(\belief_k,\action_k) &= \E\{\cost(\state_k,\action_k)+ \nlcost(\state_k,\belief_k,\action_k)| \info_k\} \\
&= \cost_{\action_k}^\p \belief_k + \nlCost(\belief_k,\action_k), \\
\text{ where } & \cost_\action = (c(\action,1),\ldots,\cost(\action,\statedim))^\p,  \\ & \nlCost(\belief_k,\action_k)  \ole \E\{\nlcost(\state_k,\belief_k,\action_k)| \info_k\} = \sum_{i=1}^\statedim {\nlcost}(i,\belief_k,\action_k)\, \belief_k(i) .\end{split}
\eeq
Define the controlled sensing objective
\begin{align}\label{eq:costcst}
J_{\policy}(\belief_0) = \Ep\left\{\sum_{\time=0}^{\infty} \discount ^{\time} \nlCost(\belief_\time, \action_\time) \right\}.
\end{align}
In controlled sensing,
the aim is to compute the optimal  stationary policy $\optpolicy:\Belief \rightarrow \actionspace$ such that
$J_{\optpolicy}(\belief_0) \leq J_{\policy}(\belief_0)$ for all $\belief_0 \in \Belief$.
Obtaining the optimal controlled sensing policy  $\optpolicy$ is equivalent to solving
 Bellman's  dynamic programming equation:
$ \optpolicy(\belief) =  \underset{\action \in \actionspace}\argmin~ Q(\belief,\action)$, $J_{\optpolicy}(\belief_0) = \valuef(\belief_0)$, where
\begin{equation}
\valuef(\belief)  = \underset{\action \in \actionspace}\min ~Q(\belief,\action), \quad
  Q(\belief,\action) =   \Cost(\belief,\action)+ \discount\sum_{\obs \in \obsdim} \valuef\left(\filter\left(\belief,\obs,\action\right)\right)\filterd \left(\belief,\obs,\action\right). \label{eq:bellmancs}
\end{equation}
\subsection{Examples of Nonlinear Cost POMDP}
The non-standard feature of the objective (\ref{eq:costcst}) is  the nonlinear performance loss terms  $\nlCost(\belief,\action)$.
 These   costs\footnote{A linear function $ \cost_\action^\p \belief $ cannot attain its maximum at the centroid
of a simplex since a linear function achieves it maximum at a boundary point.}  should be 
chosen so that they are   zero at the vertices $e_i$ of the belief space $\Belief$  (reflecting perfect state estimation) 
and largest at the centroid of the belief space (most uncertain estimate). We now discuss examples of $\nlcost(\state,\belief,\action)$ and
its conditional expectation $\nlCost(\belief,\action)$.  \\
{\em (i). Piecewise Linear Cost}: Here we choose the performance loss as
\beq \nlcost(\state, \belief,\action) = \begin{cases}
   0  & \text{ if } \|\state - \belief\|_\infty \leq \epsilon \\
     \epsilon & \text{ if } \epsilon \leq \|\state- \belief\|_\infty \leq 1 - \epsilon \\
      1 & \text{ if } \|\state-\belief\|_\infty \geq 1 - \epsilon  \end{cases} , \quad \epsilon \in [0,0.5]. \label{eq:pwcost}\eeq
        Then $\nlCost(\belief,\action)$ is piecewise linear and concave.
      This cost is useful for subjective decision making. e.g., the distance of a target to a radar is quantized into three regions: close, medium
      and far.
          \\
{\em (ii). Mean Square, $l_1$ and  $l_\infty$ Performance Loss}: Suppose in (\ref{eq:costcst}) we choose
   \begin{equation}  \label{eq:l2u} \nlcost(\state, \pi,u) =
 \alpha(u) (\state - \belief)^\p M (\state - \belief) + \beta(u)
 , \;\state \in \{e_1,\ldots,e_\statedim\}, \pi \in \Pi.
\end{equation}
Here $M$ is  a user defined positive semi-definite symmetric matrix, $\alpha(u)$ and $\beta(u)$, $u \in \actionspace$ are user defined positive scalar
weights  that
allow different sensors (sensing modes)
to be  weighed differently. So (\ref{eq:l2u})   is the 
squared error of the Bayesian estimator (weighted by $M$, scaled by $\alpha(u)$
and translated by $\beta(u)$). 
In terms of the belief state, the mean square performance loss (\ref{eq:l2u})  is 
\beq  \nlCost(\belief_k, \action_k) = \E\{\nlcost(\state_k,\belief_k,\action_k)| \info_k\}=
 \alpha(\action_k) \bigl( \sum_{i=1}^\statedim M_{ii} \belief_k(i) - \belief_k^\p M \belief_k\bigr) + \beta(u_k)
\label{eq:quadcostconcave}
\eeq
because $\E\{ (x_k-\belief_k)^\p M (x_k-\belief_k) | \info_k\} = \sum_{i=1}^\statedim (e_i - \belief)^\p M (e_i - \belief) \belief(i)$.
The cost (\ref{eq:quadcostconcave}) is quadratic and concave in the belief. \\
Alternatively, if $ \nlcost(\state, \pi,u)  = \|\state - \belief\|_1$ then  $ \nlCost(\belief, \action) = 2 (1 - \belief^\p \belief)$ is also quadratic
in the belief. Also, choosing  $ \nlcost(\state, \pi,u)  = \|\state - \belief\|_\infty$ yields $ \nlCost(\belief, \action) =  (1 - \belief^\p \belief)$.

\noindent
{\em  (iii). Entropy based  Performance Loss}:
Here   we choose
  \begin{equation}  
   \nlCost(\belief,\action)
 =-\alpha(u) \sum_{i=1}^S \pi(i) \ltwo \pi(i) + \beta(u), \qquad  \pi \in \Pi .\label{eq:entropy}\end{equation}
The intuition 
is that an inaccurate sensor with cheap usage cost yields a Bayesian  estimate $\pi$ 
with a higher entropy compared to an accurate sensor.

\section{Structural Result 1 - Convexity of Value Function and Stopping Set}
Our  first result is that the value function $V(\pi)$  in (\ref{eq:bellmancs}) is concave in $\belief \in \Belief$.

\begin{theorem} \label{thm:concavevaluegencost} Consider a POMDP with possibly continuous-valued observations.
Assume that for each action $u$, the instantaneous cost $\Cost(\belief,\action)$
 are concave and continuous with respect to  $\belief \in \Belief$.
Then the value function  $\valuef(\belief)$ is concave in $\belief$.
\end{theorem}
The proof is given in \cite[Chapter 8]{Kri16}.

\subsection{Convexity of Stopping Set for Stopping Time POMDPs with nonlinear cost}
With the above concavity result we have the following important result for contolled sensing stopping time POMDPs.
A stopping time  POMDP has  action space $\actionspace = \{1 \text{ (stop)},2 \text{ (continue)}\}$.

The  stop action  $u=1$ incurs a terminal cost
of $\cost(\state,\action=1)$ and the problem terminates.

For continue action $\action= 2$, the  state $\state \in \statespace = \{1,2,\ldots,\statedim\}$  evolves with transition matrix $\tp$ and is observed
 via observations
$\obs$ with
observation probabilities $\oprob_{\state\obs} = \prob(\obs_k=\obs|\state_k=\state)$.   An instantaneous  cost $\cost(\state,\action=2)$ is incurred. Thus for $\action=2$, 
the belief state evolves according to the HMM filter
$\belief_{k} = \filter(\belief_{k-1},\obs_k)$. Since  action 1 is  a stop action and has no dynamics, to simplify notation,
we write $\filter(\belief,\obs,2)$ as $\filter(\belief,\obs)$ and $\filterd(\belief,\obs,2)$ as $\filterd(\belief,\obs)$  in this subsection.

For the  stopping time POMDP, $\optpolicy$ is the solution of
 Bellman's equation which is of the form  
\begin{align} \label{eq:bellmanstop}
 \optpolicy(\belief) &=  \underset{\action \in \actionspace}\argmin ~\valueaction(\belief,\action), \quad 
\optvalue(\belief) = \underset{\action \in \actionspace}\min ~\valueaction(\belief,\action),     \\
  \valueaction(\belief,1) &= \cost_1^\prime\belief , \quad
  \valueaction(\belief,2) =  \Cost(\belief,2) + \discount\sum_{\obs \in \obsdim} \optvalue\left(\filter\left(\belief,\obs\right)\right)\filternorm \left(\belief,\obs\right).
\nonumber 
\end{align}
where $\filter(\belief,\obs)$ and $\filterd(\belief,\obs)$ are the HMM filter and normalization (\ref{eq:information_state}).

We now present the first  structural result for stopping time POMDPs:  the stopping region for the optimal policy is convex.
Define the stopping set $\region_1$  as the set of belief states for which stopping ($\action=1$)  is the optimal action.
Define $\region_2$ as the set of belief states for which continuing ($\action=2$) is the optimal action. That is
\beq
\region_1 = \{\belief:  \optpolicy(\belief) = 1  \text{ (stop) }\} , \quad \region_2 =  \{\belief:  \optpolicy(\belief) = 2 \} = \Belief - \region_1.\eeq

  The theorem below shows that the stopping set $\region_1$  is convex (and therefore a connected set). 
Recall that the value function $\valuef(\belief)$ is concave on $\Belief$. 

\begin{theorem} 
\label{thm:pomdpconvex}
Consider the stopping-time POMDP with value function given by (\ref{eq:bellmanstop}). Suppose that the possibly nonlinear cost $C(\pi,2)$ is concave in $\pi$.
Then the stopping set $\region_1$ is a convex subset of the belief space $\Belief$. \index{convexity of stopping region}
\end{theorem}

\begin{proof}
Pick any two belief states $\belief_1,\belief_2 \in \region_1$. To demonstrate convexity of $\region_1$,
we need to show for any $\lambda \in [0,1]$,  $\lambda \belief_1 + (1-\lambda) \belief_2 \in \region_1$.
Since $V(\belief)$ is concave,
\begin{align*}
V(\lambda \belief_1 + (1-\lambda) \belief_2) &\geq \lambda V(\belief_1) + (1-\lambda) V(\belief_2) \nonumber\\
&= \lambda Q(\belief_1,1) + (1-\lambda) Q(\belief_2,1)  \text{ (since $\belief_1,\belief_2 \in \region_1$) } \nonumber\\
&= Q(\lambda \belief_1 + (1-\lambda) \belief_2,1 ) \text{ (since $Q_{1}(\belief,1)$ is linear in $\belief$) }\nonumber \\
& \hspace{-1cm} \geq V(\lambda \belief_1 + (1-\lambda) \belief_2) \text{ (since $V(\belief)$ is the optimal value function) }
\end{align*}
Thus all the inequalities above are equalities, and $\lambda \belief_1 + (1-\lambda) \belief_2 \in 
\region_1$.
\end{proof}

The above theorem is a small extension of \cite{Lov87a} which deals with case when the costs $C(\pi,2)$ are linear in $\pi$.
The proof is exactly the same as in \cite{Lov87a} -- all that is required is that $C(\pi,2)$ is concave

\subsection{Example.  Quickest Change Detection with Nonlinear Delay Cost}  \label{sec:classicalqd}  \index{quickest detection! classical|(}

Quickest  detection is a useful example of a stopping time POMDP that  has applications in
numerous areas   \cite{PH08,BN93}. 
The classical Bayesian quickest detection problem is as follows: 
An underlying discrete-time state process $\state$ jump changes at a geometrically distributed random time $\tau^0$.
Consider a sequence of random measurements $\{\obs_k,k \geq 1\}$, such that 
 conditioned on the event $\{\tau^0 = t\}$, $\obs_k$, $\{k \leq t\}$  are i.i.d. random variables with distribution 
$\oprob_{1\obs}$ and $\{\obs_k, k >t\}$ are i.i.d. random variables with distribution $\oprob_{2\obs}$.
The quickest  detection problem involves detecting the change time $\tau^0$ with minimal cost. That is,
at each time $k=1,2,\ldots$, a decision $u_k \in \{\text{continue}, \text{stop and announce change}\}$ needs to be made to optimize a tradeoff
between false alarm frequency and linear delay penalty.\footnote{There are two general formulations for quickest time
detection.  In the first
formulation, the change point $\tau^0$ is an unknown deterministic time,
and the goal is to determine a
stopping rule such that a  worst case delay penalty is
minimized subject to a constraint on the false alarm frequency
(see, e.g., \cite{Mou86,Poo98,YKP99,PH08}). 
The second formulation, which is the formulation considered in this book (this chapter and also Chapter \ref{chp:stopapply}),  is  the  Bayesian approach where
the change time $\tau^0$ is specified by a prior distribution.}

A geometrically distributed change time $\tau^0$ is realized by a  two state ($\statedim=2$) Markov chain  
with absorbing transition matrix $\tp$ and prior $\belief_0$ as follows:
\beq \tp = \begin{bmatrix} 1 & 0 \\ 1- \tp_{22} & \tp_{22}  \end{bmatrix} , \;  \belief_0 = \begin{bmatrix} 0 \\ 1 \end{bmatrix} , \quad
\tau^0 = \inf\{ k:  \state_k = 1\}. \label{eq:tpqdp} \eeq
The system starts  in state 2 and then jumps to  the absorbing state 1 at time $\tau^0$. Clearly $\tau^0$ is geometrically distributed
with mean $1/(1-\tp_{22})$.

The cost criterion in classical quickest detection is the {\em Kolmogorov--Shiryayev 
criterion} for detection of disorder \cite{Shi63}  \index{quickest detection! Kolmogorov--Shiryayev  criterion of disorder}
 \beq J_\policy(\belief) =   d\, \Ep\{(\tau - \tau^0)^+\} + \Pp(\tau < \tau^0) , \quad \belief_0 = \belief.
\label{eq:ksd} \eeq
where $\policy$ denotes the decision policy.
The first term is the delay penalty in making a decision at time $\tau > \tau^0$ and $d$ is a positive real number.
The second term is the false alarm penalty incurred in announcing a change at time $\tau< \tau^0$.

\noindent {\bf Stopping time POMDP}:
The quickest detection problem with penalty (\ref{eq:ksd}) is a stopping time POMDP with 
$\actionspace = \{1 \text{ (announce change and stop)},2  \text{ (continue)} \}$,  $\statespace=\{1,2\}$,
transition matrix in (\ref{eq:tpqdp}), arbitrary observation probabilities $\oprob_{\state\obs}$,
 cost vectors  $c_1 = [ 0 ,\; 1 ]^\p$, $c_2 = [d,\; 0 ]^\p$ and discount factor $ \discount = 1$.

In light of Theorem  \ref{thm:pomdpconvex}, we can generalize this to delay costs  $C(\pi,2)$ that are convex and nonlinear in the belief.
For example such a cost could be motivated by the square error or entropy of the belief reflecting an inaccurate state estimate.
We have the following  structural result.
 \begin{corollary} \label{cor:qdclassical}
 The optimal policy $\optpolicy$ for classical quickest detection has a {\em threshold} structure:
There exists a threshold point $\belief^* \in [0,1]$ such that  
\beq u_k = \optpolicy(\belief_k) = \begin{cases} 2 \text{ (continue) } & \text{ if }
\belief_k(2) \in [ \belief^*,1] \\   1 \text{ (stop and announce change)  } &  \text{ if } \belief_k(2) \in [0, \belief^*).
\end{cases} \label{eq:onedim}
\eeq  \end{corollary}
\begin{proof} Since $\statedim=2$, $\Belief$ is  the interval $[0,1]$, and   $\belief(2) \in [0,1]$ is the belief state.
Theorem \ref{thm:pomdpconvex} implies that the stopping set $\region_1$ is convex. In one dimension this implies  that 
$\region_1$ is an interval of the form $[a^*,\pi^*)$ for $0 \leq a< \pi^*\leq 1$. Since state 1 is absorbing,
 Bellman's equation (\ref{eq:bellmanstop}) with $\discount=1$ applied at $\belief = e_1$  implies
$$\optpolicy(e_1) = \argmin_u\{\underbrace{\cost(1,u=1)}_{0},\;\; d ( 1 - \belief(2)) + V(e_1)\} = 1.$$ 
So $ e_1$ or equivalently $\belief(2) = 0$ belongs to $\region_1$. Therefore,
$\region_1$ is an interval of the form $[0,\belief^*)$.
Hence  the optimal policy is of the form  (\ref{eq:onedim}). \end{proof}

 Theorem \ref{thm:pomdpconvex} says that   for quickest
 detection of a multi-state Markov chain,
 the stopping set $\region_1$  is  convex for any concave  non-linear delay cost. This is different to the result in \cite{Kri11} which considered a nonlinear stopping cost (false alarm cost) - in \cite{Kri11} the stopping
 set was not necessarily convex. For additional results on controlled sampling with quickest detection see \cite{Kri13}.

 \subsection*{Social Learning}
Social learning, or learning from the actions of others, is an integral part of human behavior and has been studied widely  in behavioral economics, sociology, electrical engineering and 
 computer science
  to model the  interaction of  decision makers \cite{BHW92,AO11,Cha04,EK10,Say14b,WD16,Kri12,KP13,KP14}. POMDPs with social learning result in interesting behaviour.

Social learning models  present unique challenges from a statistical signal processing point of view.
First,  agents interact with and influence each other.  For example, ratings posted on online reputation systems strongly influence the behavior of  individuals.
  This is usually not the case with physical sensors.  
Second,  agents (humans) lack the capability to quickly absorb information  and translate it into decisions.
According to the paradigm of rational inattention theory, pioneered
by economics Nobel prize winner Sims \cite{Sim03},  attention is a time-limited resource that can be modelled in terms of an information-theoretic channel capacity. Therefore, while apparently mistaken decisions are ubiquitous, this does not imply that decision makers are irrational.\footnote{Limits on attention impact choice. For example, purchasers limit their attention to a relatively small
number of websites when buying over the internet; shoppers buy expensive products due to their failure to notice if sales tax is includes in the price  \cite{CD15}.}
  More recently for results in quickest detection POMDPs with social learning and risk averse agents
 please see \cite{Kri12,KB16}.

{\em Remark}: Of course, one of the best known examples of a stopping time problem is optimal search for a Markov target \cite{Eag84,MJ95,SK02,JK06}.
Another interesting example is a multiple stopping problem \cite{Nak95,KAB16}; this has applications in interactive advertising in social multimedia like YouTube. The problem has distinct parallels to scheduling in communication
systems \cite{NK09}.

\section{The value function is positively homogenous}

Define the positive $\statedim$-orthant as  $\reals_+^\statedim$.
On this positive orthant, define the relaxed belief state $\rbelief$. We can define the following Bellman's equation where $\valuer$ below denotes the value function with $\rbelief \in \reals_+^\statedim$.

\beq
\valuer(\rbelief)  = \underset{\action \in \actionspace}\min ~Q(\rbelief,\action), \quad
  Q(\rbelief,\action) =  ~\cost_\action^\prime\rbelief + \discount\sum_{\obs \in \obsdim} \valuer\left(\filter\left(\rbelief,\obs,\action\right)\right)\filterd \left(\rbelief,\obs,\action\right). \label{eq:rbellman}
\end{equation}

Clearly when $\rbelief $ is restricted to the belief space (unit simplex) $\Belief$, then $\valuer(\alpha) = \valuef(\alpha)$.  This can be established by mathematical induction (valued iteration) and the proof is omitted.
We now have the following result.

\begin{theorem} The relaxed value function $\valuer(\cdot)$ of a linear cost POMDP is positively homogenous. That is, for any constant $\kappa > 0$, 
$\valuer(\kappa \rbelief) =  \kappa \valuer(\rbelief)$. Therefore,  (\ref{eq:rbellman}) can be expressed as 
\beq  \label{eq:valuerep}
 \valuer(\rbelief)  = \underset{\action \in \actionspace}\min ~Q(\rbelief,\action), \quad Q(\rbelief,\action) =  ~\cost_\action^\prime\rbelief + \discount\sum_{\obs \in \obsdim} \valuer\left( \oprob_\obs(\action) \tp^\p(\action) \rbelief \right) 
\eeq

\end{theorem}
The proof is straightforward since the cost $c_u^\p \rbelief $  and $\sigma(\rbelief,\obs,\action)$ are linear in $\rbelief$ and 
$\filter(\kappa \rbelief, \obs,\mu) = \filter( \rbelief, \obs,\mu)$. 

It is this positive homogeneity property of the value function and especially the representation (\ref{eq:valuerep}) which allows for the finite horizon case to immediately show that the value function is piecewise linear
and concave.

\section{Monotone Value Function}

 \begin{definition}[Monotone Likelihood Ratio (MLR) order $\gr$]  \index{stochastic dominance! monotone likelihood ratio} \label{def:mlr}
Let $\belief_1, \belief_2 \in \Belief$ denote  two beliefs.
Then $\belief_1$ dominates $\belief_2$ with respect to the MLR order, denoted as
$\belief_1 \gr \belief_2$,
 if 
$ \belief_1(i) \belief_2(j) \leq \belief_2(i) \belief_1(j)$ $i < j$,  $i,j\in \{1,\ldots,\statedim\}$.  A function $\fun:\Belief\rightarrow \reals$ is said to be MLR increasing if $\belief_1 \gr \belief_2$ implies $\fun(\belief_1) \geq \fun(\belief_2)$. 

\end{definition}
\begin{enumerate}[\bf{(A}1)]

\item\label{it:decreasing_cost}  $\Cost(\belief,\action)$ is  first order stochastic  decreasing in $\belief$ for each $ \action \in \actionspace$.
\item\label{it:TP2_gen} $\tp(\action)$,  $\action \in \actionspace$ is  totally positive of order 2 (TP2):  all second-order minors are nonnegative.
\item\label{it:TP2_obs} $\oprob(\action)$,  $\action \in \actionspace$ is  totally positive of order 2 (TP2).
\end{enumerate}

\begin{theorem} \label{thm:valuedec}
Under A\ref{it:decreasing_cost},  A\ref{it:TP2_gen} and A\ref{it:TP2_obs}, the value function $V(\belief)$ in (\ref{eq:bellmancs}) is MLR decreasing.
\end{theorem}
The proof of the theorem is in \cite{Kri16,KD07}.

Note that for 2 states  ($\statedim = 2$),  one can  always  permute the observation labels so that A\ref{it:TP2_obs} holds. 
Moreover, A\ref{it:TP2_gen} then becomes the same as the first row being first order stochastic dominated by the second row.
Therefore for $\statedim=2$, the conditons for a monotone value function for a POMDP are identical to that for a fully observed MDP.

Based on extensive numerical experiments, we conjecture that assumption A\ref{it:TP2_obs} is not required for  Theorem \ref{thm:valuedec}.

\begin{conjecture} Under A\ref{it:decreasing_cost} and   A\ref{it:TP2_gen}, the value function $V(\belief)$ in (\ref{eq:bellmancs}) is MLR decreasing.
\end{conjecture}

This conjecture implies that monotone value functions for POMDPs require very similar conditions to monotone value functions for fully observed MDPs.
Of course, the TP2 condition  A\ref{it:TP2_gen} for the transition matrix is stronger than the first order dominance conditions on the transition matrix used for fully observed MDPs.

Finally we mention that one can also show that the value function involving controlled sensing with a Kalman filter is monotone \cite{KBGM12}.  In this case, the covariance matrices of the Kalman filters
are partially ordered with respect to positive definiteness.
Results for monotone HMM filters are given in 
\cite{KR14}. These monotone results can also be used for POMDP bandits as discussed in \cite{KW09}. One can also consider controlled sampling of an evolving duplication deletion graph; the dynamics of the belief
are given by the HMM filter as described in \cite{KBP16}.

\section{Blackwell Dominance and Optimality of Myopic Policies}

\subsection{Myopic Policy Bound to Optimal Decision Policy} \label{sec:myopic}

Motivated by active sensing applications, consider the following POMDPs where
based on the current
belief state $\belief_{k-1}$, agent $k$ chooses sensing mode $$u_k \in \{1 \text{ (low resolution sensor) }, 2 \text{ (high resolution sensor)}\}.$$

The assumption that mode $u=2$ yields more accurate observations than mode $u=1$ is modeled as follows: We say mode 2 {\em Blackwell
dominates} mode 1, denoted as \index{Blackwell dominance}
\beq \label{eq:bd}
  \oprob(2) \bd \oprob(1)\quad  \text{ if }  \quad  B(1)  = B(2)\, \aB . \eeq
   Here $\aB$ is a $Y^{(2)} \times Y^{(1)}$ stochastic matrix. $\aB$ can be viewed as a {\em confusion matrix} that maps $\obspace^{(2)}$ probabilistically to $\obspace^{(1)}$.
   (In a communications context, one can view $\aB$ as a noisy discrete memoryless channel with input $y^{(2)}$ and output $y^{(1)}$).
Intuitively (\ref{eq:bd})  means that
  $B^{(2)} $ is more accurate than   $B^{(1)} $.

The goal is to compute the optimal policy $\mu^*(\belief) \in \{1,2\}$ to minimize the  expected cumulative  cost incurred by  all the agents
\beq
J_\mu(\belief) = \Ep \{ \sum_{k=0}^\infty \discount^{k} C(\belief_{k},u_k) \} .
\eeq
where  $\discount \in [0,1)$ is  the discount factor.
Even though solving the above POMDP 
is computationally intractable in general,
 using Blackwell dominance, we show below that a myopic policy forms a lower
 bound for the optimal policy.

The
value function $V(\belief)$ and optimal policy $\mu^*(\belief)$  satisfy Bellman's equation
\beq  \label{eq:dp_algmove}\begin{split}
V(\belief) &= \min_{u \in \actionspace} Q(\belief,u), \quad
\mu^*(\belief)= \arg\min_{u \in \actionspace} Q(\belief,u) ,\; J_{\mu^*}(\belief) = V(\belief) \\
 Q(\belief,u) &=  C(\belief,u) 
+ \discount \sum_{y^{(u)} \in \obspace^{(u)}}  V\left( T(\belief ,y, {u}) \right) \sigma(\belief,y, {u}), \\
T(\belief,y, {u}) &= \frac{B_{y^{(u)}} (\action) P^\p \belief}{\sigma(\belief,\obs,\action)},
\; \filterd(\belief,y, \action) = \mathbf{1}_X^\p B_{y^{(u)}}(\action) P^\p \belief .
\end{split} \eeq

We now present the structural result.
Let $\Pi^s \subset \I$ denote the set of belief states for which
$C(\belief,2) < C(\belief,1)$. 
Define the  myopic policy
$$\policyl(\belief) = \begin{cases} 2 & \belief \in \Pi^s \\
 								1 & \text{ otherwise } \end{cases}$$

\begin{theorem}\label{thm:compare2}  Assume that $\Cost(\belief,\action)$ is concave with respect to $\belief \in \Belief$
for each action $\action$.
Suppose  $\oprob(2) \bd \oprob(1)  $, i.e., $    B(1)  = B(2) \aB $ holds where $\aB$ is a stochastic matrix.
Then the myopic policy
 $\policyl(\belief)$ is a lower    bound to the optimal
policy $\mu^*(\belief)$, i.e.,  $\mu^*(\belief) \geq \policyl(\belief)$ for
all $\belief \in \I$. 
In particular, 
for $\belief \in \Pi^s$,  $\mu^*(\belief) = \policyl(\belief)$,
i.e., it is optimal to choose action 2 when the  belief is in $\Pi^s$.
\qed
\end{theorem}

{\em Remark}: If  $\oprob(1) \bd \oprob(2)  $, then the myopic policy constitutes an upper bound to the optimal policy.

Theorem \ref{thm:compare2} is proved below.
 The proof exploits the fact that the value function is concave  and uses Jensen's inequality.
The usefulness of Theorem \ref{thm:compare2} stems from the fact that $\policyl(\belief)$ is trivial to compute. It  forms a provable
lower bound to the computationally intractable optimal policy $\mu^*(\belief)$.  
Since $\policyl$ is sub-optimal, it incurs a higher   cumulative cost. This  cumulative cost can be evaluated via simulation and is
an upper bound to the achievable  optimal cost.

Theorem \ref{thm:compare2} is  non-trivial.
The instantaneous costs
satisfying $C(\belief,2) < C(\belief,1)$,  does not trivially  imply that the myopic policy 
$\policyl(\belief)$ coincides with the optimal policy $\mu^*(\belief)$, since the optimal policy applies to a cumulative cost function involving
an infinite horizon
 trajectory of the dynamical system.


\subsection{Example 1.  Optimal Filter  vs Predictor Scheduling} \index{sensor scheduling! filter vs predictor} 
\index{Blackwell dominance! filter vs predictor scheduling}
Suppose $\action=2$ is an active sensor  (filter) which obtains measurements of the underlying Markov chain and uses the optimal HMM  filter
on these measurements to compute the belief and therefore
the state estimate.
 So the usage cost of sensor 2 is high (since obtaining observations is expensive and can also result in increased threat of being discovered), but its performance cost is
low (performance quality is high). 

Suppose  sensor $\action=1$ is a predictor which needs no measurement. So its usage cost  is low (no measurement is required).
However its performance cost is high since it is more inaccurate compared to sensor 2.

Since the predictor has non-informative observation probabilities, its observation probability matrix is $\oprob(1) = 
\frac{1}{\obsdim}\ones_{\statedim \times \obsdim}$. So clearly
$\oprob(1) = \oprob(2)  \oprob(1) $
meaning that the filter (sensor 2)  Blackwell dominates the predictor (sensor 1)
Theorem \ref{thm:compare2} then says that if the current belief is $\belief_k$, then if  $\Cost(\belief_k,2) < \Cost(\belief_k,1)$, it is always
optimal to deploy the filter (sensor 2).

\subsection{Example 2.  Ultrametric Matrices  and Blackwell Dominance}
\index{Blackwell dominance! ultrametric matrix|(}
An $\statedim\times \statedim$ square matrix $\oprob$ is   a symmetric stochastic ultrametric matrix if \index{ultrametric matrix}
\index{root of stochastic matrix} \index{stochastic matrix! roots}
\index{Blackwell dominance! ultrametric matrix}
\begin{compactenum}
\item $\oprob$ is symmetric and stochastic. 
\item $\oprob_{ij} \geq \min \{\oprob_{ik}, \oprob_{kj}\}$ for all $i,j,k \in \{1,2,\ldots,\statedim\}$.
\item $\oprob_{ii} > \max \{\oprob_{ik}\}, k \in  \{1,2,\ldots,\statedim\}- \{i\}$  (diagonally dominant).
\end{compactenum}
 It is shown in \cite{HL11}  that if $\oprob$ is a symmetric  stochastic  ultrametric matrix, then the $\actiondim$-th root, namely $\oprob^{1/\actiondim}$, is also a stochastic matrix\footnote{Although we do not pursue it here, conditions that 
 ensure that  the $\actiondim$-th root of a transition matrix is a valid stochastic matrix  is important
 in interpolating Markov chains. For example, transition matrices for credit ratings on a yearly time scale
 can be obtained from rating agencies such as Standard \& Poor's. Determining the transition matrix for periods of six months involves the square root of the yearly transition matrix \cite{HL11}.}
for any positive  integer $\actiondim$. Then with $\bd$ denoting Blackwell dominance (\ref{eq:bd}),
clearly  $$\oprob^{1/\actiondim} \bd \oprob^{2/(\actiondim)} \bd \cdots \bd \oprob^{(\actiondim-1)/\actiondim} \bd \oprob . $$
Consider a  social network where
the reputations of agents are denoted as $\action \in \{1,2,\ldots, \actiondim\}$. An agent with reputation $\action$
has  observation probability matrix $\oprob^{(\actiondim-\action+1)/\actiondim}$.
So an agent with reputation 1 (lowest reputation) is $\actiondim$ degrees of separation from the source signal while an agent with reputation $\actiondim$ (highest reputation) is  \index{ultrametric matrix! Blackwell dominance}
1 degree of separation from the source signal. \index{social network! sampling}
The underlying source (state) could be a news event, sentiment or corporate policy that evolves
with time.  A marketing agency can sample these agents - it can sample high reputation agents that have accurate
observations but this costs more than sampling low reputation agents that have less accurate observations. 
Then Theorem \ref{thm:compare2}  gives a suboptimal policy that provably lower bounds the optimal sampling policy.
\index{Blackwell dominance! ultrametric matrix|)}

\subsection{Proof of Theorem  \ref{thm:compare2}}

Recall from Theorem \ref{thm:concavevaluegencost} 
 that $\Cost(\belief,\action)$ concave implies that   $V(\belief)$ is concave on $\Belief$. 
We then use the Blackwell dominance condition (\ref{eq:bd}). In particular,
\begin{align*} \filter(\belief,\yi,1) &=   \sum_{\yii \in \obspace^{(2)}} \filter(\belief,\yii,2) \frac{\filterd(\belief,\yii,2)}{\sigs(\belief,\yi,1)} P(\yi|\yii)  \\
 \sigs(\belief,\yi,1) &= \sum_{\yii \in \obspace^{(2)}} \filterd(\belief,\yii,2) P(\yi|\yii). \end{align*}
Therefore $\frac{\filterd(\belief,\yii,2)}{\sigs(\belief,\yi,1)} P(\yi|\yii) $ is a probability measure w.r.t.\  $\yii$ (since the denominator is the sum
of the numerator over all $\yii$).
Since $V(\cdot)$ is concave, using Jensen's inequality it follows that
\begin{align}
&V(\filter(\belief,\yi,1) )  = V \left(\sum_{\yii \in \obspace^{(2)}} \filter(\belief,\yii,2) \frac{\filterd(\belief,\yii,2)}{\sigs(\belief,\yi,1)} P(\yi|\yii) \right)\nn  \\
&\geq \sum_{\yii \in \obspace^{(2)}}  V (\filter(\belief,\yii,2)) \frac{\filterd(\belief,\yii,2)}{\sigs(\belief,\yi,1)} P(\yi|\yii) \nn \\
&\implies   \sum_{\yi}  V(\filter(\belief,\yi,1) ) \sigs(\belief,\yi,1) \geq
\sum_{\yii} V(\filter(\belief,\yii,2)\filterd(\belief,\yii,2). \label{eq:bdproof1}
\end{align}
Therefore for $\belief \in \Pi^s$, 
$$ C(\belief,2) + \discount\sum_{\yii} V(\filter(\belief,\yii,2)\filterd(\belief,\yii,2) \leq 
C(\belief,1) + \discount \sum_{\yi}  V(\filter(\belief,\yi),1 ) \sigs(\belief,\yi,1) . $$
So for $\belief \in \Pi^s$, the optimal policy $\mu^*(\belief) = \arg\min_{u \in \actionspace}Q(\belief,u) = 2$.
So $\policyl(\belief) = \mu^*(\belief)=2$ for $\belief \in \Pi^s$ and $\bar{\mu}(\belief)=1$ otherwise, implying that
$\bar{\mu}(\belief)$ is a lower  bound for $\mu^*(\belief)$.

The above result is quite general and can be extended to  controlled sensing of  jump Markov linear systems \cite{DGK01,LK99,EKN05}.

\section{Inverse POMDPs and Revealed Preferences}
{\em How to develop data-centric non-parametric methods (algorithms and associated mathematical analysis) to identify utility functions of agents?}
Classical statistical decision theory  arising in electrical engineering (and statistics) is model based: given a model\footnote{In non-parametric detection theory, the set of decision rules can be considered to be the model.}, we wish to detect specific events in a dataset.  The  goal is the reverse:
given a dataset, we wish to determine  if 
the actions of  agents are consistent with  utility maximization behavior, or more generally, consistent with play from a  Nash equilibrium; and  we then wish to 
estimate the associated utility function. Such problems will be studied using  revealed preference methods arising in  micro-economics. Classical revealed preferences deals with
analyzing choices made by individuals.  The celebrated  ``Afriat's theorem" \cite{Var82,Blu05} provides a necessary and sufficient condition for a finite dataset  to have originated from a utility maximizer.
Specifically,  revealed preferences~\cite{KH15,HNK15,HKA16}, rational inattention, homophily~\cite{NHK16}, and social learning can be used to study 
multi-agent behavior in  social networks; particularly YouTube.


%
%
%







\bibliographystyle{plain}
\bibliography{$HOME/styles/bib/vkm}

\end{document}